\documentclass[conference]{IEEEtran}
\usepackage{bm} 
\usepackage{float}
\usepackage{graphicx}
\usepackage{support-caption}
\usepackage{subcaption}
\usepackage{amsmath}
\usepackage{algorithm}
\usepackage{algorithmicx}
\usepackage{algpseudocode}
\usepackage{filecontents}
\usepackage{comment}
\usepackage{enumerate}
\usepackage{amsfonts}
\usepackage{amssymb}
\usepackage{verbatim}
\usepackage{amsthm}
\usepackage{setspace}
\usepackage{mathrsfs}
\newcommand\figref{Figure~\ref}

\usepackage{etoolbox}

\newtheorem{theorem}{Theorem}

\newtheorem{lemma}[theorem]{Lemma}
\newtheorem{proposition}{Proposition}
\newtheorem{definition}{Definition}

\makeatletter

\newcommand{\Rmnum}[1]{\expandafter\@slowromancap\romannumeral #1@}
\makeatother

\begin{document}
\title{Resilient Distributed Diffusion for Multi-task Estimation}

\author{
\IEEEauthorblockN{
   Jiani Li and$^{}$            
   Xenofon Koutsoukos$^{}$       
}

\IEEEauthorblockA{
  $^{}$Institute for Software Integrated Systems\\
  Vanderbilt University\\
  $^{}$$\{\text{jiani.li, xenofon.koutsoukos}\}$@vanderbilt.edu
}
}


%


\maketitle

\begin{abstract}
Distributed diffusion is a powerful algorithm for multi-task state estimation which enables networked agents to interact with neighbors to process input data and diffuse information across the network. 
Compared to a centralized approach, diffusion offers multiple advantages that include robustness to node and link failures. 
In this paper, we consider distributed diffusion for multi-task estimation where networked agents must estimate distinct but correlated states of interest by processing streaming data.
By exploiting the adaptive weights 
used for diffusing information, we develop attack models that drive normal agents to converge to states selected by the attacker.
The attack models can be used for both stationary and non-stationary state estimation.
In addition, we develop a resilient distributed diffusion algorithm under the assumption that the number of compromised nodes in the neighborhood of each normal node is bounded by $F$ and we show that resilience may be obtained at the cost of performance degradation. Finally, we evaluate the proposed attack models and resilient distributed diffusion algorithm using stationary and non-stationary multi-target localization. 
\end{abstract}

\section {Introduction}

Diffusion Least-Mean Squares (DLMS) is a powerful algorithm for distributed state estimation \cite{journals/spm/SayedTCZT13}.
It enables networked agents to interact with neighbors to process
streaming data and diffuse information across the network to continually perform the estimation tasks. 
Compared to a centralized approach, diffusion offers multiple advantages that include robustness to drifts 
in the statistical properties of the data, scalability, relying on local data and 
fast response among others. 
Applications of distributed diffusion include spectrum sensing in cognitive networks \cite{7086338}, 
target localization \cite{targetLocalization}, distributed clustering \cite{6232902},  
and biologically inspired designs for mobile networks \cite{mobileAdaptiveNetworks}. 

Diffusion strategies have been shown to be robust to node and link failures
as well as nodes or links with high noise levels \cite{5948418, 7472545}. 
Resilience of diffusion-based distributed algorithms in the presence of intruders has 
been studied in \cite{6232902,journals/spm/SayedTCZT13,yuanchen/AdversaryDetection}. 
The main idea is to use adaptive weights to counteract the attacks. 

In this paper, we consider distributed diffusion for multi-task estimation where networked agents must estimate distinct but correlated states of interest by processing streaming data.
We are interested in understanding if adaptive weights 
introduce vulnerabilities that can be exploited by an attacker.
The first problem we consider is to analyze if it is possible for an attacker 
to compromise a node so that it can make nodes in the neighborhood of the compromised 
node  converge to a state selected by the attacker. 
Then, we consider a network attack and we want to determine which minimum set of nodes  
to compromise in order to make the entire network to converge to states 
selected by the attacker.
Our final objective is to design a resilient distributed diffusion 
algorithm to protect against  attacks and continue the operation possibly 
with a degraded performance.  
We do not rely on detection methods to improve resilience because distributed detection with only local information may lead to false alarms \cite{6545301}.

Distributed optimization and estimation can be performed also using consensus algorithms. 
Resilience of consensus-based distributed algorithms in the presence of cyber attacks
has received considerable attention \cite{5779706, DBLP:journals/jsac/LeBlancZKS13,6736104, 7822915}. 
Typical approaches usually assume Byzantine faults and consider that the goal of the attacker
is to disrupt the convergence (stability) of the distributed algorithm. In contrast, this 
paper focuses on attacks that do not disrupt convergence but drive the normal agents to
converge to states selected by the attacker.

The contributions of the paper are:
\begin{enumerate}
\item By exploiting the adaptive weights used for diffusing information, we develop 
attack models that drive normal agents to converge to states selected by an attacker.
The attack models can be used for deceiving a specific node or the entire network and 
apply to both stationary and non-stationary state estimation.

\item We develop a resilient distributed diffusion algorithm under the assumption that 
the number of compromised nodes in the neighborhood of each normal node is bounded by 
$F$ and we show the resilience may be obtained at the cost of performance degradation. 
If the parameter $F$ selected by the normal agents is large, then the resilient distributed diffusion algorithm degenerates to noncooperative estimation.

\item We evaluate the proposed attack models and the resilient estimation algorithm 
using both stationary and non-stationary multi-target localization. 
The simulation results are consistent with our theoretical analysis and show that 
the approach provides resilience to attacks while incurring performance degradation
which depends on the assumption about the number of nodes that has been compromised. 
\end{enumerate}

The paper is organized as follows: Section \Rmnum{2} briefly introduces distributed
diffusion. Section \Rmnum{3} presents the attack and resilient distributed diffusion problems. 
The single node and network attack models are presented in Section \Rmnum{4} and \Rmnum{5} respectively. Section \Rmnum{6}, presents and analyzes the resilient distributed diffusion algorithm. Section \Rmnum{7} presents simulation results for evaluating the approach for
multi-target localization. Section \Rmnum{8} overviews related work and Section \Rmnum{9}
concludes the paper.

\section{Preliminaries}\label{preliminaries}
We use normal and boldface font to denote deterministic and random variables respectively.
The superscript $(\cdot)^{*}$ denotes complex conjugation for scalars and complex-conjugate 
transposition for matrices, $\mathbb{E}\{\cdot\}$ denotes expectation,  
and $\|\cdot\|$ denotes the euclidean norm of a vector.

Consider a connected network of $N$ (static) agents. 
At each iteration $i$, each agent $k$ has access to a scalar measurement 
$\bm{d}_{k}(i)$ and a regression vector $\bm{u}_{k,i}$ of size $M$ with 
zero-mean and uniform covariance matrix 
$R_{u, k} \triangleq \mathbb{E}\{\bm{u}_{k,i}^* \bm{u}_{k,i}\} > 0$, 
which are related via a linear model of the following form:
\begin{equation*}\label{eq:1}
\bm{d}_{k}(i) = \bm{u}_{k,i} w_k^0 + \bm{v}_{k}(i)
\end{equation*}
where $ \bm{v}_{k}(i)$ represents a zero-mean i.i.d. additive noise 
with variance $\sigma^2_{v,k}$ and
$w_k^{0}$ denotes the unknown $M\times 1$ state vector of agent $k$. 

The objective of each agent is to estimate $w_k^{0}$ from (streaming) data
$\{\bm{d}_{k}(i),\bm{u}_{k,i}\} $ $(k=1,2,...,N, i \geq 0)$. 
The model can be static or dynamic and we represent the objective state as $w_{k}^0$ or 
$\bm{w}^0_{k,i}$ respectively. 
÷≥For simplicity, we use $w_k^0$ to denote the objective state in both the static and dynamic case. 

The state $w_k^0$ can be computed as the the unique minimizer of the following cost function:
\begin{equation*}
J_k(w) \triangleq \mathbb{E} \{ \|\bm{d}_{k}(i)- \bm{u}_{k,i}w\|^2\}
\end{equation*}
An elegant adaptive solution for determining $w_k^0$ is the least-mean-squares (LMS) filter
\cite{journals/spm/SayedTCZT13}, 
where each agent $k$ computes successive estimators of $w_k^0$ without cooperation 
(noncooperative LMS) as follows:
\begin{equation*}
\bm{w}_{k,i} = \bm{w}_{k,i-1} + \mu_k \bm{u}_{k,i}^*[\bm{d}_{k}(i)-\bm{u}_{k,i}\bm{w}_{k,i-1}]
\end{equation*}

Compared to noncooperative LMS, diffusion strategies introduce an aggregation 
step that incorporates into the adaptation mechanism information collected from other
agents in the local neighborhood. 
One powerful diffusion scheme is adapt-then-combine (ATC) \cite{journals/spm/SayedTCZT13} 
which optimizes the solution in a distributed and adaptive way using the following update:
\begin{align*}
& \bm{\psi}_{k,i} =op& \text{(adaptation)}\\
& \bm{w}_{k,i} = \sum_{l \in \mathcal{N}_k} a_{lk}(i) \bm{\psi}_{l,i} & \text{(combination)}
\end{align*}
where $\mathcal{N}_{k}$ denotes the neighborhood set of agent $k$ 
including $k$ itself, 
$\mu_k>0$ is the step size (can be identical or distinct across agents), 
$a_{lk}(i)$ represents the weight assigned to agent $l$ from agent $k$  
that is used to scale the data it receives from $l$, 
and the weights satisfy the following constraints:
\begin{equation*}
 a_{lk}(i)\geq0,  \qquad \sum_{l \in \mathcal{N}_k} a_{lk}(i) = 1, \qquad a_{lk}(i)=0 \text{ if } l\not\in \mathcal{N}_k.
\end{equation*}

In the case when the agents estimate a common state 
$w^0$ (i.e., $w_k^0$ is the same for every $k$), 
several combination rules can be adopted such as Laplacian, Metropolis, averaging, 
and maximum-degree \cite{DBLP:journals/corr/abs-1205-4220}. 
In the case of multiple tasks, the agents are pursuing distinct but 
correlated objectives $w_k^0$. In this case, the combination rules mentioned above are not 
applicable because they simply combine the estimation of all neighbors without 
distinguishing if the neighbors are pursuing the same objective. An agent
estimating a different state will prevent its neighbors from estimating the 
state of interest.

Diffusion LMS (DLMS) has been extended for multi-task networks in \cite{6232902}
using the following adaptive weights:
\begin{equation}\label{eq: adaptive relative-variance combination rule}
a_{lk}(i)=
\begin{cases}
\frac {\gamma_{lk}^{-2}(i)} {\sum_{m \in \mathcal{N}_k}\gamma_{mk}^{-2}(i) }, & l \in \mathcal{N}_k\\
0, & \text{otherwise}
\end{cases}
\end{equation}
where $\gamma_{lk}^2(i) = (1-\nu_k)\gamma_{lk}^2(i-1)+\nu_k \| \bm{\psi}_{l,i}-\bm{w}_{k,i-1}\| ^2$ and  
$\nu_k$ is a positive step size known as the forgetting factor.
This update enables the agents to continuously 
learn which neighbors should cooperate with and which should not. 
During the estimation task, agents pursuing different objectives will assign 
to each other continuously smaller weights according to 
\eqref{eq: adaptive relative-variance combination rule}. 
Once the weights become negligible, the communication link between the agents 
does not contribute to the estimation task. As a result, as the estimation proceeds, 
only agents estimating the same state cooperate.

DLMS with adaptive weights (DLMSAW) outperforms the noncooperative LMS as measured by the 
steady-state mean-square-deviation performance (MSD) \cite{journals/spm/SayedTCZT13}. 
For sufficiently small step-sizes, the network performance of noncooperative LMS is 
defined as the average MSD level:  
\begin{equation*}
\text{MSD}_{\text{ncop}} \triangleq \lim_{i \rightarrow \infty} \frac{1}{N} \sum_{k=1}^N \mathbb{E} \| \tilde{\bm{w}}_{k,i}\|^2 \approx \frac{\mu M}{2} \cdot (\frac{1}{N} \sum_{k=1}^N \sigma_{v,k}^2)
\end{equation*}
where $\tilde{\bm{w}}_{k,i} \triangleq w_k^0 - \bm{w}_{k,i}$.
The network MSD performance of the diffusion network (as well as the MSD performance of a normal agent in the diffusion network) can be approximated by 
\begin{equation*}
\text{MSD}_{\text{k}} \approx \text{MSD}_{\text{diff}} \approx \frac{\mu M}{2} \cdot \frac{1}{N} \cdot (\frac{1}{N} \sum_{k=1}^N \sigma_{v,k}^2)
\end{equation*}
In \cite{journals/spm/SayedTCZT13}, it is shown that $\text{MSD}_{\text{diff}} = \frac{1}{N} \text{MSD}_{\text{ncop}}$, which demonstrates an $N$-fold improvement of MSD performance.

\section{Problem formulation}
Diffusion strategies have been shown to be robust to node and link failures
as well as nodes or links with high noise levels \cite{5948418, 7472545}. 
In this paper, we are interested in understanding if the adaptive weights 
provide resilience in the case a subset of networked nodes is compromised by 
cyber attacks.
The first problem being considered is to analyze if it is possible for an attacker 
to compromise a node so that it can make nodes in the neighborhood of this node 
converge to a state selected by the attacker. 
Then, we consider a network attack model to determine which minimum set of nodes  
to compromise in order to make the entire network to converge to states 
selected by the attacker.
Finally, we would like to design a resilient distributed 
algorithm to protect against attacks and continue the operation possibly 
with a degraded performance.

\subsection{Single Node Attack Model}
We consider false data injection attacks, and thus attacks only incur between neighbors 
exchanging messages. We assume that 
the attacker(s) know the topology of the network, the streaming data received by each agent, 
and the parameters used by the agents (e.g., $\mu_k$).
Compromised nodes are assumed to be Byzantine in the sense that they can send arbitrary 
messages to their neighbors, and also they can send different messages to different neighbors.
The objective of the attacker is to drive the normal nodes to converge to
a specific state. 
We assume a compromised node $a$ wants agent $k$ to converge to state 
\begin{equation*}
w_{k,i}^a=
\begin{cases}
w_k^a, &\text{for stationary estimation}\\
w_k^a + \theta_{k,i}^a, &\text{for non-stationary estimation}
\end{cases}
\end{equation*}
We define the objective function of the attacker as
\begin{equation}\label{eq: objective function}
\min_{\bm{w}_{k,i}}  \|\bm{w}_{k,i} - w^a_{k,i}\|, \qquad i \rightarrow \infty, \qquad \bm{w}^a_{k,i} \in D_{w,k}.
\end{equation}
where $D_{w,k}$ is the domain of state $\bm{w}_{k,i}$.


Another objective of the attacker can be to delay the convergence time of the normal agents.
One observation is that if the compromised node can make its neighbors to converge to a selected state, it can keep changing this state before normal neighbors converge. By doing so, normal neighbors being attacked will never converge to a fixed state. And thus, the attacker can achieve its goal to prolong the convergence time of normal neighbors. 
For that reason, we focus on the attack model based on objective \eqref{eq: objective function}. 

\subsection{Network Attack Model}

Determining which nodes to compromise 
is another problem. If the attacker has a specific target node that she wants to attack and make
it converge to a specific state, the attacker can compromise any neighbors of this node in order to achieve the objective. In the case the attacker wants to compromise the entire network and drive the multi-task estimation to specific states,  she needs to find a minimum set of nodes that will enable the attack in order to compromise the least possible nodes.

\subsection{Resilient Distributed Diffusion}\label{attack detection}
Distributed diffusion is said to be \emph{resilient} if 
\begin{equation}\label{eq: resilient distributed diffusion}
   \lim_{i \rightarrow \infty} \bm{w}_{k,i} = w_k^0 
\end{equation}
for all normal agents $k$ in the network which
ensures that all the noncompromised nodes will converge to the true state. 
We assume that in the neighborhood of a normal node, there could be at most $F$ compromised 
nodes~\cite{DBLP:journals/jsac/LeBlancZKS13}. 
Assuming bounds on the number of adversaries is typical for security and resilience of 
distributed algorithms.
We consider the problem of modifying DLMSAW to achieve resilience
while possibly incurring a performance degradation as measured by the MSD level.

\section{Single Node Attack Design}\label{the section of attack model}

In order to achieve the objective \eqref{eq: objective function}, a compromised node $a$ can send messages to
a neighbor node $ k $ so that the adaptive weights are assigned such that the state $\bm{w}_{k,i}$ (estimated by $k$)  is driven to $w_{k,i}^a$.
We assume the attack starts at $i_a \geq 0$ and the attack succeeds if $\exists i_c$, s.t. $\forall i > i_c$, $\| \bm{w}_{k,i} - w^a_{k,i}\| < \epsilon$, for some small value $\epsilon > 0$.

\begingroup
\makeatletter
\apptocmd{\thetheorem}{\unless\ifx\protect\@unexpandable@protect\protect\footnote{Proofs can be found in the Appendix.}\fi}{}{}
\makeatother

\begin{lemma}\label{mainthm}
If a compromised node $a$ wants to make a normal neighbor $k$ converge to a selected state 
$ w_{k,i}^a $, 
then it should follow a strategy to make the weight assigned by $k$ satisfy:

\textit{1. Stationary estimation:} 
$\exists  i_a'  \geq i_a$, s.t. $ (\forall i_a' \leq i \leq i_c, \forall l \in \mathcal{N}_k \cap l \neq a$, $a_{ak}(i) \gg a_{lk}(i) )$ $\wedge$  $ \neg (\forall i_a' \leq i \leq i_c, $ $ \forall \epsilon > 0, a_{ak}(i) > 1 - \epsilon) $ $\wedge$ $ (\forall i > i_c, \forall \epsilon >0, a_{ak}(i) > 1-\epsilon)$.

\textit{2. Non-stationary estimation:} $\exists  i_a'  \geq i_a$, s.t. $(\forall i \geq i_a', $ $ \forall l \in \mathcal{N}_k \cap l \neq a$, $a_{ak}(i) \gg a_{lk}(i) )$ $\wedge$ $ \neg (\forall i \geq i_a', $ $ \forall \epsilon > 0, a_{ak}(i) > 1 - \epsilon) $.




\end{lemma}
\endgroup

A compromised node can implement the attack by manipulating the value of  $\psi_{a,k,i}$
to satisfy \textit{Lemma 1}. \textit{Lemma 2} presents a sufficient condition for selecting
$\psi_{a,k,i}$ that satisfy the attack strategy in \textit{Lemma 1}.

\begin{lemma}
The strategy in \textit{Lemma 1} can be satisfied by selecting $\psi_{a,k,i}$ to satisfy the
following conditions:

\textit{1. Stationary estimation:}
$\forall l \in \mathcal{N}_k \cap l \neq a$, $(\forall i_a \leq i \leq i_c, $ $ \|\psi_{a,k,i} - \bm{w}_{k,i-1}\| \ll \|\psi_{l,k} - \bm{w}_{k,i-1}\|) $ $ \wedge $ $ \neg (\forall i_a \leq i \leq i_c, $ $ \|\psi_{a,k,i} - \bm{w}_{k,i-1}\| = 0)$ $ \wedge $ $ (\forall i > i_c, \|\psi_{a,k,i} - \bm{w}_{k,i-1}\| = 0)$.


\textit{2. Non-stationary estimation:} 
$\forall l \in \mathcal{N}_k \cap l \neq a$, $(\forall i \geq i_a, $ $ \|\psi_{a,k,i} - \bm{w}_{k,i-1}\| \ll \|\psi_{l,k} - \bm{w}_{k,i-1}\|) $ $ \wedge $ $ \neg (\forall i \geq i_a, $ $ \|\psi_{a,k,i} - \bm{w}_{k,i-1}\| = 0)$.

\end{lemma}

For a compromised node to send a message to its normal neighbors satisfying the conditions in \textit{Lemma 2}, it needs to compute $\bm{w}_{k,i-1}$.

\begin{lemma}\label{lemma: attacker can deduce w_k,i-1}
If a compromised node $a$ has knowledge of node $k$'s streaming data $\{\bm{d}_k(i), \bm{u}_{k,i}\}$ and 
the parameter $\mu_k$, then it can compute 
$\bm{w}_{k,i-1}$. 
\end{lemma}

Based on \textit{Lemma 2}, 
$\psi_{a,k,i} = \bm{w}_{k,i} + \Delta_{k,i}$, $\text{for some small } \| \Delta_{k,i}\| \geq 0$.
For stationary state estimation, we can select $\Delta_{k,i} = r_{k,i}^a (w_{k}^a - \bm{w}_{k,i-1})$, 
where $r_{k,i}^a$ is a small coefficient representing the step size, 
and $w_{k}^a - \bm{w}_{k,i-1}$ is the steepest slope vector towards $w_{k}^a$ at state $\bm{w}_{k,i}$.
When $\bm{w}_{k,i-1}$ converges to $w_{k}^a$, we have $\| \Delta_{k,i}\| = 0$ and thus  
$\psi_{a,k,i} = \bm{w}_{k, i}$, satisfying the condition for $i > i_c$.
For non-stationary state estimation, if $\Delta_{k,i} = r_{k,i}^a (w_{k,i}^a - \bm{w}_{k,i-1})$ 
then the state may converge to a state very close to $w_{k,i}^a$ but not $w_{k,i}^a$ exactly. 
Therefore, we propose the following attack model:
\begin{equation} \label{eq: attacker model}
{\psi}_{a,k,i} = \bm{w}_{k,i-1} + r_{k,i}^a (x_i - \bm{w}_{k,i-1})
\end{equation}
where $x_i$ is given by
\begin{equation*}
x_i=
\begin{cases}
w_k^a, &\text{for stationary estimation}\\
w_k^a + \theta_{k,i-1}^a + \frac{\Delta \theta_{k,i-1}^a}{r_{k,i}^a}, &\text{for non-stationary estimation}
\end{cases}
\end{equation*}
with $\Delta \theta_{k,i}^a = \theta_{k,i+1}^a - \theta_{k,i}^a$. 
The step size $r_{k,i}^a$ 
should be selected to satisfy \textit{Lemma 2}.
The following proposition provides a condition on $r_{k,i}^a $ that ensures 
the attack will achieve its objective.

\begin{proposition}\label{proposition: constraint of r_{k,i}^a}
If $r_{k,i}^a \geq 0$ is selected such that 
$\forall l \in \mathcal{N}_k \cap l \neq a$,
$(\forall i \geq i_a, $ $ \| r_{k,i}^a (x_i - \bm{w}_{k,i-1})\| \ll \| \bm{\psi}_{l,i} - \bm{w}_{k,i-1}\|)$ $ \wedge $ $ \neg (\forall i \geq i_a, r_{k,i}^a = 0)$, 
then the compromised node $a$ can realize the objective \eqref{eq: objective function} by using
${\psi}_{a,k,i}$ described in \eqref{eq: attacker model} as the communication message with $k$.
\end{proposition}
Note that for a fixed value $r_{k,i}^a$, 
it is possible that $\| r_{k,i}^a (x_i - \bm{w}_{k,i-1})\| \ll \| \bm{\psi}_{l,i} - \bm{w}_{k,i-1}\|$ 
does not hold for some iteration $i$ because of the randomness of variables. 
Yet we can always set $r_{k,i}^a = 0$ for such iterations $i$.   
However, in practice, the attack can succeed by using a small fixed value of $r_{k,i}^a > 0$. 
The reason may be that because of the smoothing property of the weight, estimation is robust to 
infrequent small values of $\| \bm{\psi}_{l,i} - \bm{w}_{k,i-1}\|$ caused by randomness.

\section{Network Attack Design}\label{the section of network attack model}
In this section, we consider the case when there are multiple compromised nodes 
using the attack model presented above.
Our objective is to determine the minimum set of nodes to compromise in order 
to attack the entire network. 
It should be noted that there is no need for multiple compromised nodes $a_1, a_2, \ldots$ 
to attack a single normal node $k$ in their neighborhood. 
The reason is that if each compromised node sends the same message to node $k$, 
we can consider only one node with 
$a_{ak}(i) = a_{a_1,k}(i) + a_{a_2,k}(i) + \ldots $, 
and design the attack using only $a_{ak}(i)$.

First, we investigate if a compromised node could indirectly impact its neighbors' neighbors.
Consider the case when node $k$ is connected to a compromised node $a$ and a normal node $l$, 
and $a$ is not connected to $l$. 
Without loss of generality, we set $\nu_k = 1$ and we use $\bm{R}_1$ and $\bm{R}_2$ 
to denote the two random variables $\mu_k \bm{u}_{k,i}^* e_{k}(i)$ and $\mu_l \bm{u}_{l,i}^* e_{l}(i)$. 
Then, for $i > i_a$, the weight assigned to node $k$ by node $l$ is given by
\begin{equation}\label{eq: weight lk}
a_{kl}(i) = \frac{\|\bm{w}_{k,i-1} + \bm{R}_1 - \bm{w}_{l,i-1}\|^{-2}}{\|\bm{w}_{k,i-1} + \bm{R}_1 - \bm{w}_{l,i-1}\|^{-2} + \|\bm{R}_2\|^{-2}}
\end{equation}

Suppose the compromised node $a$ could affect nodes beyond its neighborhood, 
for $i > i_a + n$, $\bm{w}_{k,i}$ converges to $w_k^a$ and $\bm{w}_{l,i}$ converges to $w_l^a$. 
Equation \eqref{eq: weight lk} can be written as
\begin{equation}\label{eq:akl}
    a_{kl}(i) = \frac{\|\bm{R}_2\|^{2}}{\|w_k^a + \bm{R}_1 - w_l^a\|^{2} + \|\bm{R}_2\|^{2}}
\end{equation}
and we have
\begin{equation}\label{eq:wla}
   w_l^a = a_{kl}(i) (w_k^a + \bm{R}_1) + (1-a_{kl}(i)) (w_l^a + \bm{R}_2)
\end{equation}
From \eqref{eq:akl} and \eqref{eq:wla}, we obtain
\begin{equation}\label{eq: 9}
        \frac{\|\bm{R}_2\|^{2}}{\|w_k^a + \bm{R}_1 - w_l^a\|^{2} + \|\bm{R}_2\|^{2}} (w_l^a - w_k^a + \bm{R}_2 - \bm{R}_1) = \bm{R}_2
\end{equation}
Since $\frac{\|\bm{R}_2\|^{2}}{\|w_k^a + \bm{R}_1 - w_l^a\|^{2} + \|\bm{R}_2\|^{2}}$ and $(\bm{R}_2 - \bm{R}_1)$ are random variables, and $(w_l^a - w_k^a)$ is a constant, \eqref{eq: 9} does not hold unless both $a_{kl}(i) = 0$ and $\bm{R}_2 = 0$. In this case, $a_{ll}(i) = 1$ and $\mu_l \bm{u}_{l,i}^* e_{l}(i) = 0$, which means $k$ does not affect $l$ and $l$ will converge to its true state.

Since a compromised node cannot affect nodes beyond its neighborhood, 
finding the minimum set of nodes to compromise in order to attack the entire network 
is equivalent to finding a minimum dominating set of the network \cite{DBLP:journals/dam/HedetniemiLP86}. 
It should be noted that finding a minimum dominating set 
of a network is an NP-complete problem but approximate solutions using greedy approaches work very well \cite{DBLP:journals/dam/HedetniemiLP86}. 


\section{Resilient Distributed Diffusion}
\subsection{Resilience Analysis}
The cost function for a normal agent $k$ at iteration $i$ is:
\begin{equation*}
    \begin{split}
J_k(\bm{w}_{k,i}) &= J_k(\sum_{l \in \mathcal{N}_k} a_{lk}(i) \bm{\psi}_{l,i}) \\
&= \mathbb{E} \{\|\bm{d}_k(i) - \bm{u}_{k,i}\sum_{l \in \mathcal{N}_k} a_{lk}(i) \bm{\psi}_{l,i}\|^2\} \\
&= \mathbb{E} \{ \|\sum_{l \in \mathcal{N}_k} a_{lk}(i) (\bm{d}_k(i) - \bm{u}_{k,i} \bm{\psi}_{l,i})\|^2\}\\
&= \sum_{l \in \mathcal{N}_k} a_{lk}^2(i) J_k(\bm{\psi}_{l,i})
    \end{split}
\end{equation*}
Obviously, the cost of $k$ is related to its neighbors' assigned weights and cost. Since $a_{lk}^2(i) J_k(\bm{\psi}_{l,i})  \propto  \frac{J_k(\bm{\psi}_{l,i})}{\gamma_{l,k}^4(i)}$, we define the contribution of $l$ to its neighbor $k$'s cost $J_k(\bm{w}_{k,i})$ as
\begin{equation*}
    c_{lk}(i) =   \frac{J_k(\bm{\psi}_{l,i})}{\gamma_{l,k}^4(i)}
\end{equation*}
To compute the cost $J_k(\bm{\psi}_{l,i}) = \mathbb{E} \|\bm{d}_k(i) - \bm{u}_{k,i} \bm{\psi}_{l,i}\|^2$, agent $k$ has to store all the streaming data. 
Alternatively, we can approximate $J_k(\bm{\psi}_{l,i})$ using a moving average based on
the previous iterations.

We assume that a normal node has at most $F$ neighbors that are  compromised 
nodes~\cite{DBLP:journals/jsac/LeBlancZKS13}. 
Specifically, we define:
\begin{definition}\label{F-local definition}
{\rm{($F$-local model)}} A node satisfies the $F$-local model if there is at most $F$ compromised nodes in its neighborhood.
\end{definition}
In general, normal nodes can select different values of $F$. 
While the paper focuses on the $F$-local model, bounds on the global  number of adversaries or 
bounds that consider the connectivity of the network are possible~\cite{DBLP:journals/jsac/LeBlancZKS13}. 

Given the $F$-local assumption, node $k$ has at most $F$ neighbors that may be compromised. 
Motivated by the W-MSR algorithm~\cite{DBLP:journals/jsac/LeBlancZKS13}, 
we modify DLMSAW as follows: 
\begin{enumerate}
\item If $F \geq |\mathcal{N}_k|$, agent $k$ updates its current state $\bm{w}_{k,i}$ using only its own $\bm{\psi}_{k,i}$, which degenerates distributed diffusion to non-cooperative LMS.
\item If $F < |\mathcal{N}_k|$, agent $k$ at each iteration $i$ computes $c_{lk}(i)$ for 
$l \in \mathcal{N}_k \text{ and } l \neq k$, sorts the results, 
and computes the set of  nodes $\mathcal{R}_k(i)$ 
consisting of $l$ for the $F$ largest $c_{lk}(i)$. Then, the agent
updates its current weight $a_{lk}(i)$ and state $\bm{w}_{k,i}$ without using information obtained 
from nodes in $\mathcal{R}_k(i)$.
\end{enumerate}
The proposed resilient distributed diffusion algorithm is summarized in \textit{Algorithm 1}.

\renewcommand{\algorithmicrequire}{\textbf{Set}}
    \begin{algorithm}\small
        \begin{algorithmic}[1] 
             \Require $ \gamma_{lk}^2(-1)=0$ , maintain $n \times 1$ matrix $D_{k,i} = \bm{0}_{n \times 1}$ and $n \times M$ matrix $U_{k,i} = \bm{0}_{n \times M}$, for all $k=1,2,...,N$, and $l \in \mathcal{N}_k$ 
             \ForAll {$k=1,2,...,N,  i \geq 0$}
             \State $e_{k}(i) = \bm{d}_{k}(i) -\bm{u}_{k,i}\bm{w}_{k,i-1}$
             \State $\bm{\psi}_{k,i}=\bm{w}_{k,i-1}+\mu_k \bm{u}_{k,i}^* e_{k}(i)$
             \If{$ F \geq |\mathcal{N}_k|$}
             \State $\bm{w}_{k,i} = \bm{\psi}_{k,i}$
             \Else
             \State $\gamma_{lk}^2(i) = (1-\nu_k)\gamma_{lk}^2(i-1)+\nu_k \| \bm{\psi}_{l,i}-\bm{w}_{k,i-1}\| ^2, l \in \mathcal{N}_k$
             \State Update $D_{k,i}$ and $U_{k,i}$ by adding $\bm{d}_k(i)$ and $\bm{u}_{k,i}$ and removing $\bm{d}_k(i-n)$ and $\bm{u}_{k,i-n}$
             \State $J_k(\bm{\psi}_{l,i}) = \mathbb{E}\|D_{k,i} - U_{k,i} \bm{\psi}_{l,i}\|^2, l \in \mathcal{N}_k$
             \State $c_{lk}(i) = \frac{J_k(\bm{\psi}_{l,i})}{\gamma_{l,k}^4(i)}, l \in \mathcal{N}_k$
             \State Sort $c_{lk}(i)$, get $\mathcal{R}_k(i)$ consisting of $l$ for the $F$ largest $c_{lk}(i)$
             \State ${a}_{lk}(i) = \frac {\gamma_{lk}^{-2}(i)} {\sum_{m \in \mathcal{N}_k \backslash R_k(i)}\gamma_{mk}^{-2}(i) }, l \in \mathcal{N}_k \backslash \mathcal{R}_k(i)$
             \State $\bm{w}_{k,i} = \sum_{l \in N_{k} \backslash \mathcal{R}_k(i)} {a}_{lk}(i) \bm{\psi}_{l,i}$
             \EndIf
            \EndFor
        \end{algorithmic}
        \caption{\small{Resilient distributed diffusion under $F$-local bounds}}
    \end{algorithm}

\begin{proposition}
If the number of compromised nodes satisfies the $F$-local model, then 
\textit{Algorithm 1} is resilient to any message falsification 
byzantine attack which aims at making normal nodes converge to a selected state.
\end{proposition}
\begin{proof}
Given the $F$-local model, there are at most $F$ neighbors of a normal agent $k$ that are compromised. In the case of $F \geq |\mathcal{N}_k|$, $k$ updates the state without using information from neighbors. 
Next, consider the case when $F < |\mathcal{N}_k|$. The algorithm removes the $F$ largest cost contributions. 
Based on the proof of \textit{Lemma 1}, we have that only for $i$ subject to $\forall l \in \mathcal{N}_k \cap l \neq a, a_{ak}(i) \gg a_{lk}(i)$, node $k$ makes progress to converge to attacker’s selected state (or stays the current state), rendering $ a_{lk}(i)  \rightarrow 0$. As a result, $c_{lk}(i) \rightarrow 0$ and thus $c_{ak}(i) \gg c_{lk}(i)$. For each iteration $i$, any compromised node $a \in \{a_1, a_2, \ldots\}$ that drives $k$ toward $w^a_{k,i}$ must be within $\mathcal{R}_k(i)$ and the message from which will be discarded. Thus,
\begin{equation*}
    \bm{w}_{k, i} = \sum_{l \in \mathcal{N}_k \backslash \mathcal{R}_k(i)} a_{lk}(i) \bm{\psi}_{l,i}
\end{equation*}
meaning the algorithm performs the diffusion adaptation as if there were no compromised node. 
Note that messages from normal neighbors may be discarded since $F$ may be greater than the number of compromised neighbors. 
However, the distributed diffusion algorithm is robust to node and link failures, 
and it converges to the true state despite the links to some or all of its neighbors fail. 
Finally, the algorithm will converge and equation \eqref{eq: resilient distributed diffusion} holds, 
showing the resilience of the \textit{Algorithm 1}.
\end{proof}

\subsection{Attacks against Resilient Distributed Diffusion}
If the number of compromised nodes satisfies the $F$-local model, 
\textit{Algorithm 1} is resilient to message falsification byzantine attacks aiming at 
driving normal nodes converge to a selected state. An important question is if there are attacks
against resilient distributed diffusion. 
The attacker could try to make the messages it sends to normal nodes not being discarded 
but affecting the convergence of normal agents. 
This must be achieved by selecting $c_{ak}(i)$ not to be one of the $F$ largest values
and thus be smaller than the value of some normal neighbor of $k$. 
In this case, $J_k(\bm{w}_{k,i})$ is even smaller than when this value is discarded
but the attacker's goal is to maximize $J_k(\bm{w}_{k,i})$. Thus, the optimal strategy for
the attacker is not to contribute cost less than a normal neighbor of $k$, and as a result, 
the information from a compromised node will be discarded.

\subsection{MSD Performance Analysis}
Each normal node must select the parameter $ F $ in order to perform resilient diffusion. 
However, if $F$ is large there will be performance degradation as measured by the MSD.
In the following, we summarize the trade-off between MSD performance and resilience.

\textit{Algorithm 1} cannot ensure resilience if $F$ is selected less than the number 
of compromised nodes in one normal agent's neighborhood. In such cases, messages from 
compromised nodes may not be entirely removed. 
However, 
as we increase $F$, the MSD level will increase.
Consider a network without compromised nodes with $N$ normal agents running 
\mbox{\textit{Algorithm 1}}.  Let $\{\sigma_{v,1}^2, \ldots, \sigma_{v,k}^2, \ldots, \sigma_{v,N}^2\}$ be the noise variance. Each agent $k$ removes the message coming from $l \in \mathcal{R}_k(i)$. 
Suppose there is a normal agent $n$, which happens to be in $\mathcal{R}_k(i)$ for every agent $k$ 
in the network at every iteration. In this case, the network will be divided into two sub-networks:
The first will consist of all the agents in the original network excluding agent $n$ and the 
second will consist of $n$ itself. The MSD of the first sub-network is
\begin{equation*}
\text{MSD}_{\text{sub1}} \approx \frac{\mu M}{2} \cdot \frac{1}{(N-1)^2} (\sum_{k=1}^N \sigma_{v,k}^2 - \sigma_{v,n}^2)
\end{equation*}
while the MSD of the second sub-network is
\begin{equation*}
\text{MSD}_{\text{sub2}} \approx \frac{\mu M}{2} \cdot \sigma_{v,n}^2.
\end{equation*}
The MSD of the entire network is
\begin{equation*}
\text{MSD}_{\text{network}}^{\text{resilient}} \approx \frac{\mu M}{2} \cdot (\frac{1}{(N-1)N} (\sum_{k=1}^N \sigma_{v,k}^2 - \sigma_{v,n}^2) + \frac{1}{N} \sigma_{v,n}^2)
\end{equation*}
The MSD of the network performing the original diffusion algorithm is given by
\begin{equation*}
\text{MSD}_{\text{network}}^{\text{original}} \approx \frac{\mu M}{2} \cdot (\frac{1}{N^2} \sum_{k=1}^N \sigma_{v,k}^2)
\end{equation*}
and the difference can be expressed as
\begin{equation*}
\begin{split}
    &\quad \text{MSD}_{\text{network}}^{\text{resilient}} - \text{MSD}_{\text{network}}^{\text{original}} \\
    &\approx \frac{\mu M}{2} \cdot (\frac{1}{N^2(N-1)} \sum_{k=1}^N \sigma_{v,k}^2 +  \frac{N-2}{N(N-1)} \sigma_{v,n}^2) > 0
\end{split}
\end{equation*}
$\text{MSD}_{\text{network}}^{\text{resilient}}$ is always larger than $\text{MSD}_{\text{network}}^{\text{original}}$, meaning the estimation performance of \mbox{\textit{Algorithm 1}} is worse than the original diffusion algorithm. 
As $F$ is increased, agents are more likely to cut links with most of their normal neighbors 
and are likely to be divided into separate sub-networks. In the worst case, agents discard all the information from their neighbors and perform the estimation tasks only using their own data. 
In this case, the algorithm will degenerate to noncooperative estimation and incur an $N$-fold MSD performance deterioration. 

\section{Evaluation}
We first evaluate the proposed attack model using a multi-target localization problem for both  stationary and non-stationary targets. 
We then evaluate the proposed resilient algorithm for stationary estimation (we omit non-stationary estimation because of length limitations). 

The network with $N=100$ agents 
is shown in \figref{fig: initial stationary network topology}. 
For stationary target localization, the coordinates of the two stationary targets are given by
\begin{equation*}
w_{k}^0=
\begin{cases}
[0.1, 0.1]^\top, & \text{for } k \text{ depicted in blue}\\
[0.9, 0.9]^\top, & \text{for } k \text{ depicted in green}
\end{cases}
\end{equation*}

 
If the weights between agents $k$ and $l$ are 
such that $a_{lk}(i) < 0.01$ and $a_{kl}(i) < 0.01$, 
the link between them is deleted. 
Regression data is white Gaussian with diagonal covariance matrices 
$R_{u,k} = \sigma_{u,k}^2 I_M$, $\sigma_{u,k}^2 \in [0.8, 1.2]$
and noise variance $\sigma_k^2 \in [0.1, 0.2]$.
The step size $\mu_k = 0.01$ and the forgetting factor $\nu_k = 0.01$ are set 
uniformly across the network. 

\figref{fig:without attack, after simulation} shows the network topology at the end of the simulation 
using DLMSAW with no attack. 
Only the links between agents estimating the same target are kept, illustrating the robustness of DLMSAW to multi-task networks.
The MSD level of the network for DLMSAW and noncooperative LMS is shown in 
\figref{fig: MSD level for stationary targets}, indicating the MSD performance improves by cooperation.

\subsection{Attack model}
\noindent \textbf{Stationary targets}: 
Suppose there are four agents in the network that are compromised by an attacker.
Compromised nodes deploy attacks on all of their neighbors using the attack model 
described in \eqref{eq: attacker model}. 
Attack parameters are selected uniformly across the compromised agents as 
$w_k^a = [0.5, 0.5]^\top$ and $r_{k,i}^a = 0.002$.
\figref{fig:attacked after simulation} shows the network topology at the end of the simulation 
(compromised nodes are red with yellow center, and normal agents converging to $w_k^a$ are denoted in red nodes). We find all the neighbors of the four compromised nodes have been successfully driven to converge to $w_k^a$,
have cut down all the links with their normal neighbors, and communicate only with the compromised nodes. 
Normal agents not communicating with the compromised nodes will end up converging to their desired targets, illustrating the conclusion in section \Rmnum{5}. 
\figref{fig: stationary_convergence_trends} shows the convergence of nodes affected by compromised nodes. 
The MSD level for DLMSAW under attack shown in \figref{fig: MSD level for stationary targets} is very high, whereas the MSD level for noncooperative LMS is not affected by the attack. 

\noindent \textbf{Non-stationary targets}: 
We assume targets with dynamics given by
\begin{spacing}{0.4}
\begin{equation*}
\bm{w}_{k,i}^0 =
\begin{cases}
\begin{bmatrix}
    0.1 + 0.1 \cos(2\pi\omega i) \\
    0.1 + 0.1 \sin(2\pi\omega i)
\end{bmatrix}
, \text{for } k \text{ depicted in blue} \\
\\
\begin{bmatrix}
    0.9 + 0.1 \cos(2\pi\omega i) \\
    0.9 + 0.1 \sin(2\pi\omega i)
\end{bmatrix}
, \text{for } k \text{ depicted in green}
\end{cases}
\end{equation*}
\\
\end{spacing}
\noindent where $\omega = \frac{1}{2000}$.
The attack parameters are selected uniformly across the compromised agents as $w_k^a = [0.5, 0.5]^\top, r_{k,i}^a = 0.002$ and $\theta_{k,i}^a = [0.1 \cos(2\pi\omega_a i),$ $0.1 \sin(2\pi\omega_a i)]^\top$, $\Delta \theta_{k,i}^a = [-0.2 \pi \omega_a \sin(2\pi\omega_a i),$ $0.2 \pi \omega_a \cos(2\pi\omega_a i)]^\top$, where $\omega_a = \frac{1}{2000}$. The attacked network topology at the end of the simulation is the same as in \figref{fig:attacked after simulation}. 
\figref{fig: non-stationary_convergence_trends} shows the average state dynamics of the neighbors of the compromised nodes. For clarity, we only show the state for the first 1900 iterations. We find that by 500 iterations, neighbors of compromised nodes have already converged to $w_{k,i}^a = w_k^a + \theta_{k,i}^a$. \figref{fig:non-stationary MSD level} shows the MSD level.

\subsection{Resilient Diffusion}
Compromised nodes are selected as described above. The cost $J_k(\bm{\psi}_{l,i})$ is approximated using the last $100$ iterations' streaming data. 
$F$ is selected by each normal agent as the expected number of compromised neighbors (We adopt uniform $F$ here but it can be distinct for each normal agent). 
For $F = 1$, the network topology  at the end of the simulation is shown in \figref{fig:network_Flocal}
illustrating the resilience of the algorithm.
The MSD level of the network for noncooperative LMS and the resilient algorithm for  $F = 0, 1, \ldots, 5$ is shown in \figref{fig: MSD_Flocal}. When $F = 0$, the algorithm is the same as the original DLMSAW, which is not resilient to attacks and has a large MSD level. 
Since each normal agent has at most one compromised node neighbor, by selecting $F = 1$ the algorithm is resilient and has a low MSD level. By increasing $F$, the algorithm is still resilient, but the MSD level increases as well, and gradually approaches the MSD level of noncooperative LMS.

\begin{figure*}[htbp]
\vspace{0cm} 
\setlength{\abovecaptionskip}{0.6cm}  
\begin{minipage}[t]{0.33\linewidth}
\centering
\includegraphics[width=0.9\textwidth, trim=1.5cm 1.5cm 1.5cm 1.5cm]{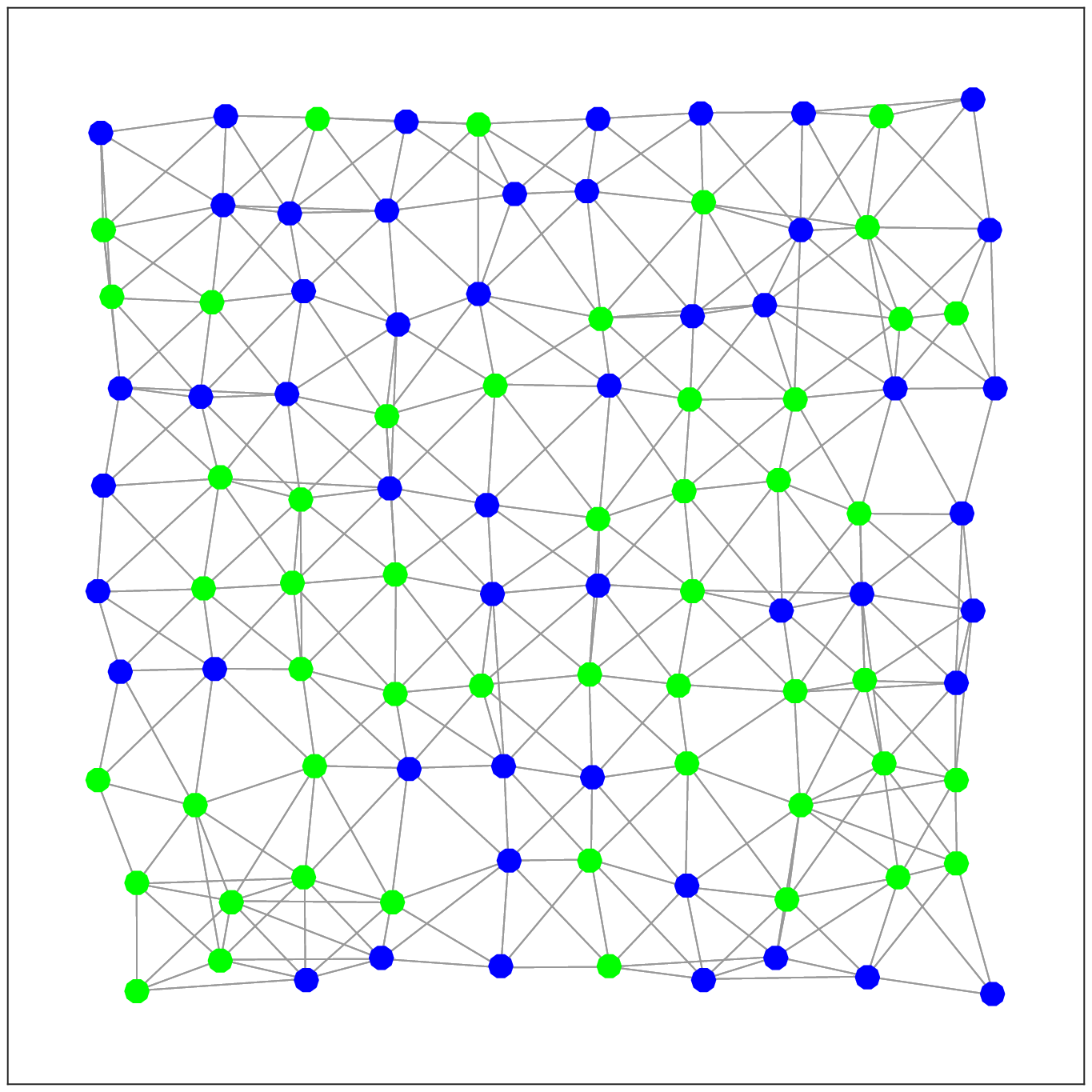}
\caption{Initial network topology}\label{fig: initial stationary network topology}
\end{minipage}%
\begin{minipage}[t]{0.33\linewidth}
\centering
\includegraphics[width=0.9\textwidth,  trim=1.5cm 1.5cm 1.5cm 1.5cm]{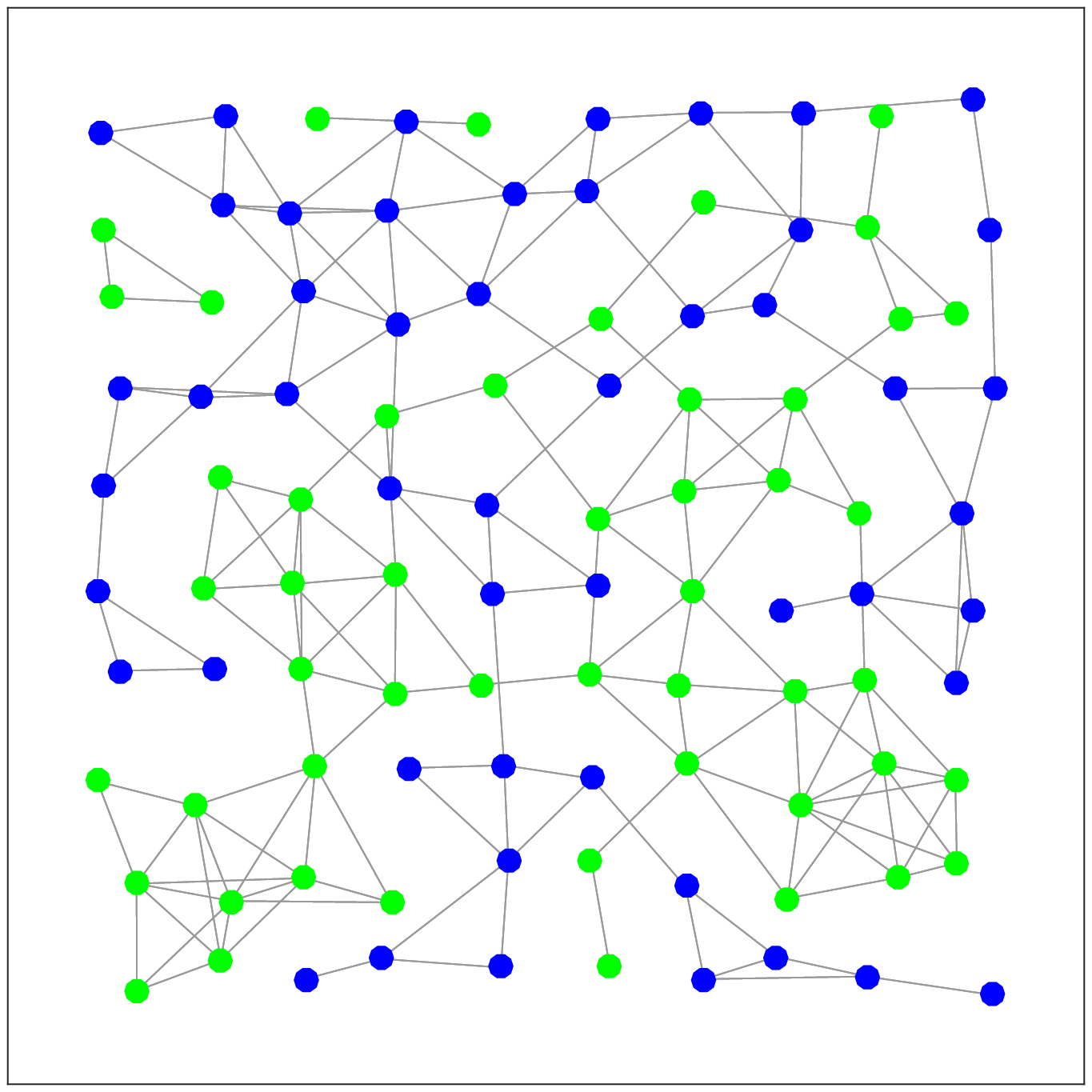}
\caption{Network topology at the end of the simulation running DLMSAW with no attack}\label{fig:without attack, after simulation}
\end{minipage}%
\begin{minipage}[t]{0.33\linewidth}
\centering
\includegraphics[width=0.9\textwidth, trim=1cm 1cm 1cm 1cm]{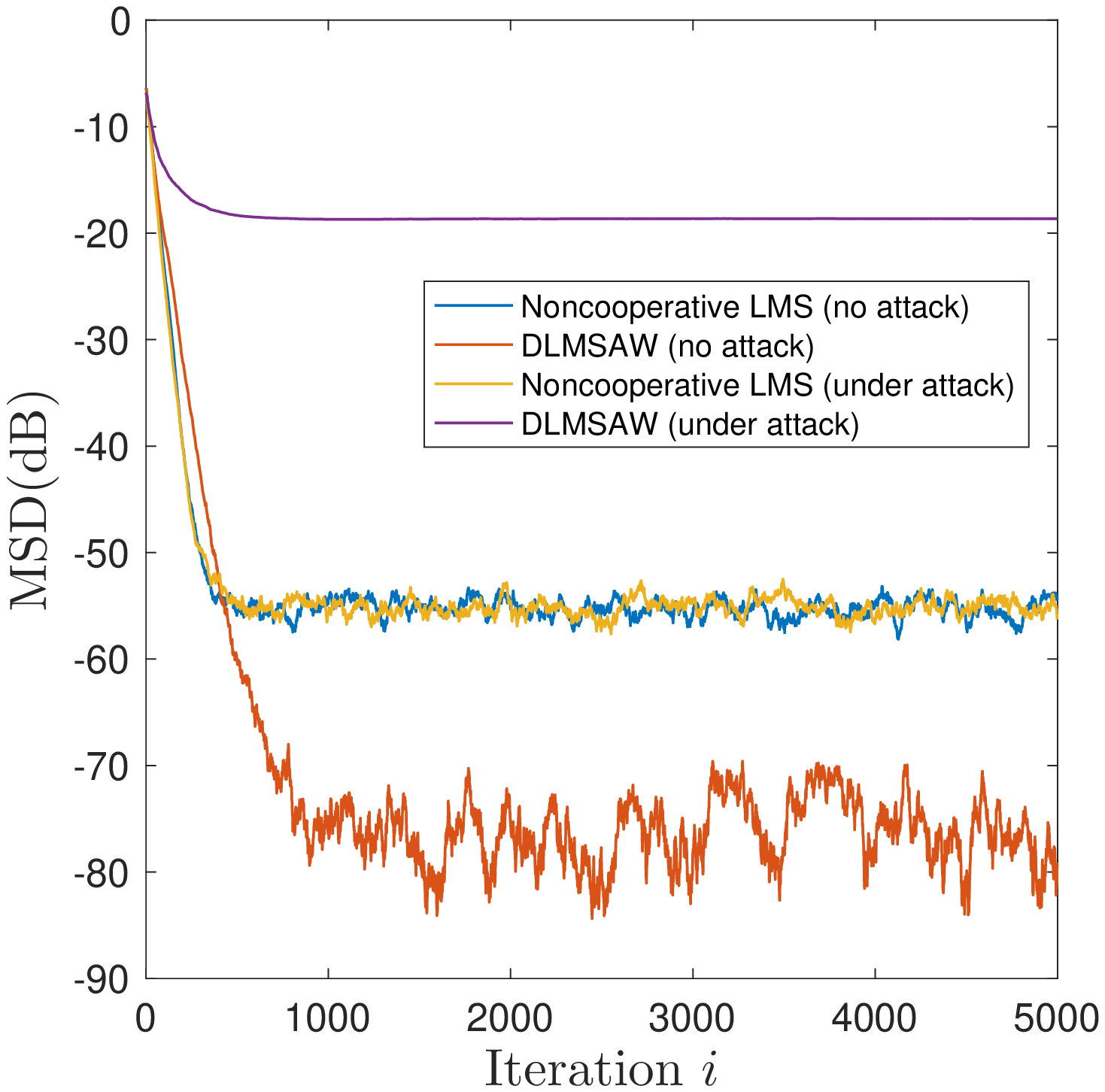}
\caption{MSD level for noncooperative LMS and DLMSAW (stationary targets)}\label{fig: MSD level for stationary targets}
\end{minipage}%
\end{figure*} 

\begin{figure*}[htbp]
\vspace{0cm} 
\setlength{\abovecaptionskip}{0.6cm}  
\begin{minipage}[t]{0.33\linewidth}
\centering
\includegraphics[width=0.9\textwidth, trim=1.5cm 1.5cm 1.5cm 1.5cm]{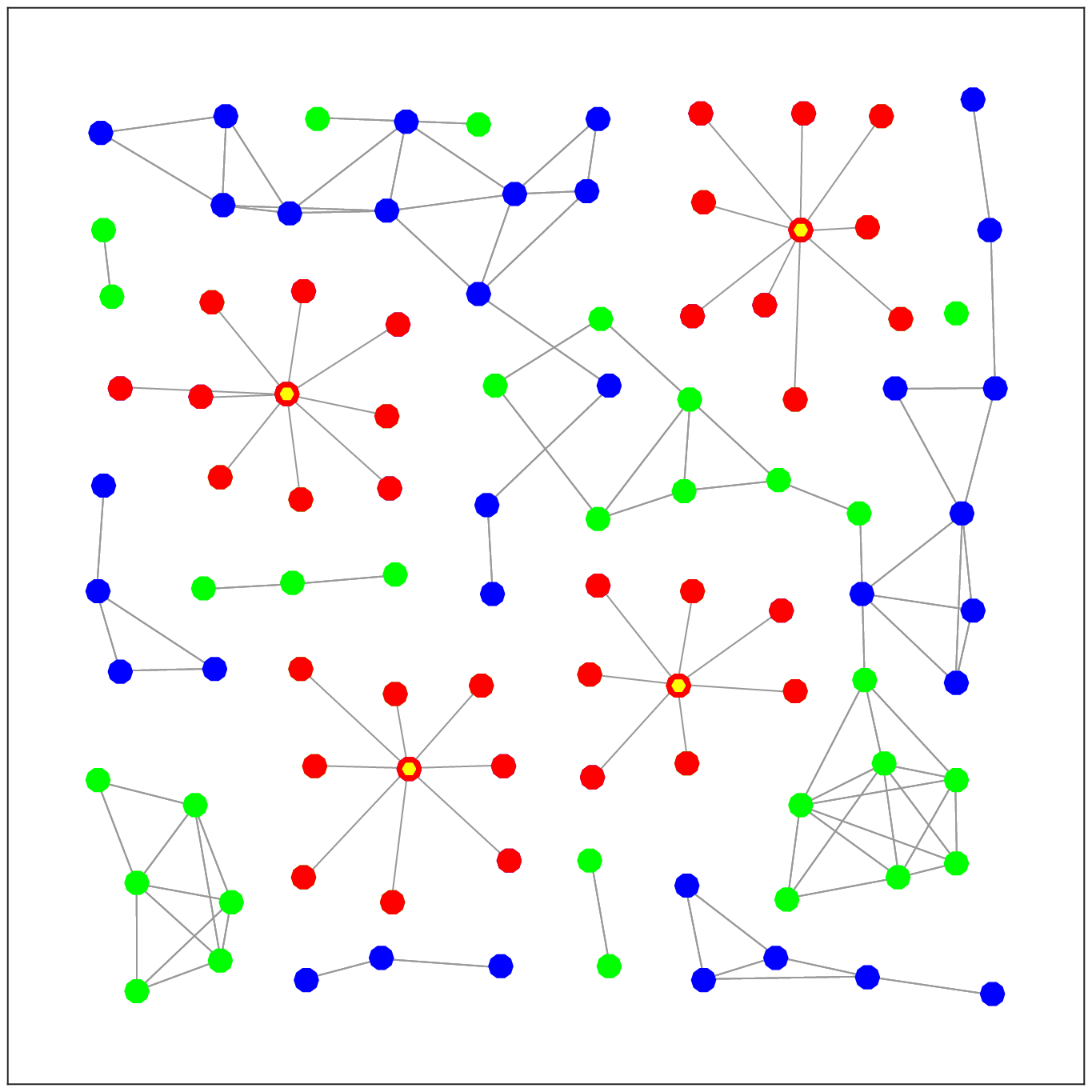}
\caption{Network topology at the end of the simulation running DLMSAW under attack}\label{fig:attacked after simulation}
\end{minipage}%
\begin{minipage}[t]{0.33\linewidth}
\centering
\includegraphics[width=0.9\textwidth, trim=1cm 1cm 1cm 1cm]{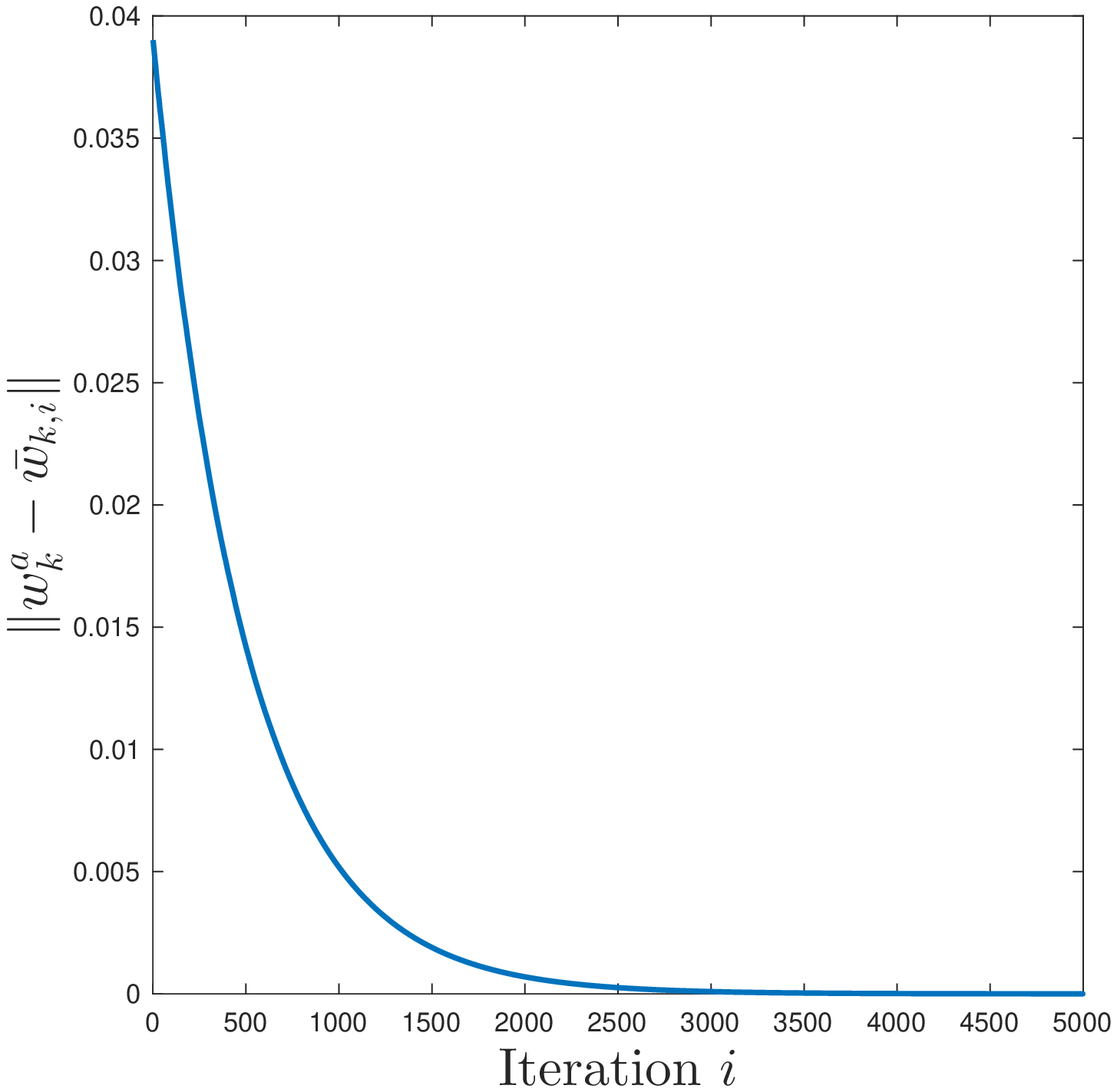}
\caption{Average state dynamics of compromised nodes' neighbors (stationary targets)}\label{fig: stationary_convergence_trends}
\end{minipage}%
\begin{minipage}[t]{0.33\linewidth}
\centering
\includegraphics[width=0.9\textwidth, trim=1cm 1cm 1cm 1cm]{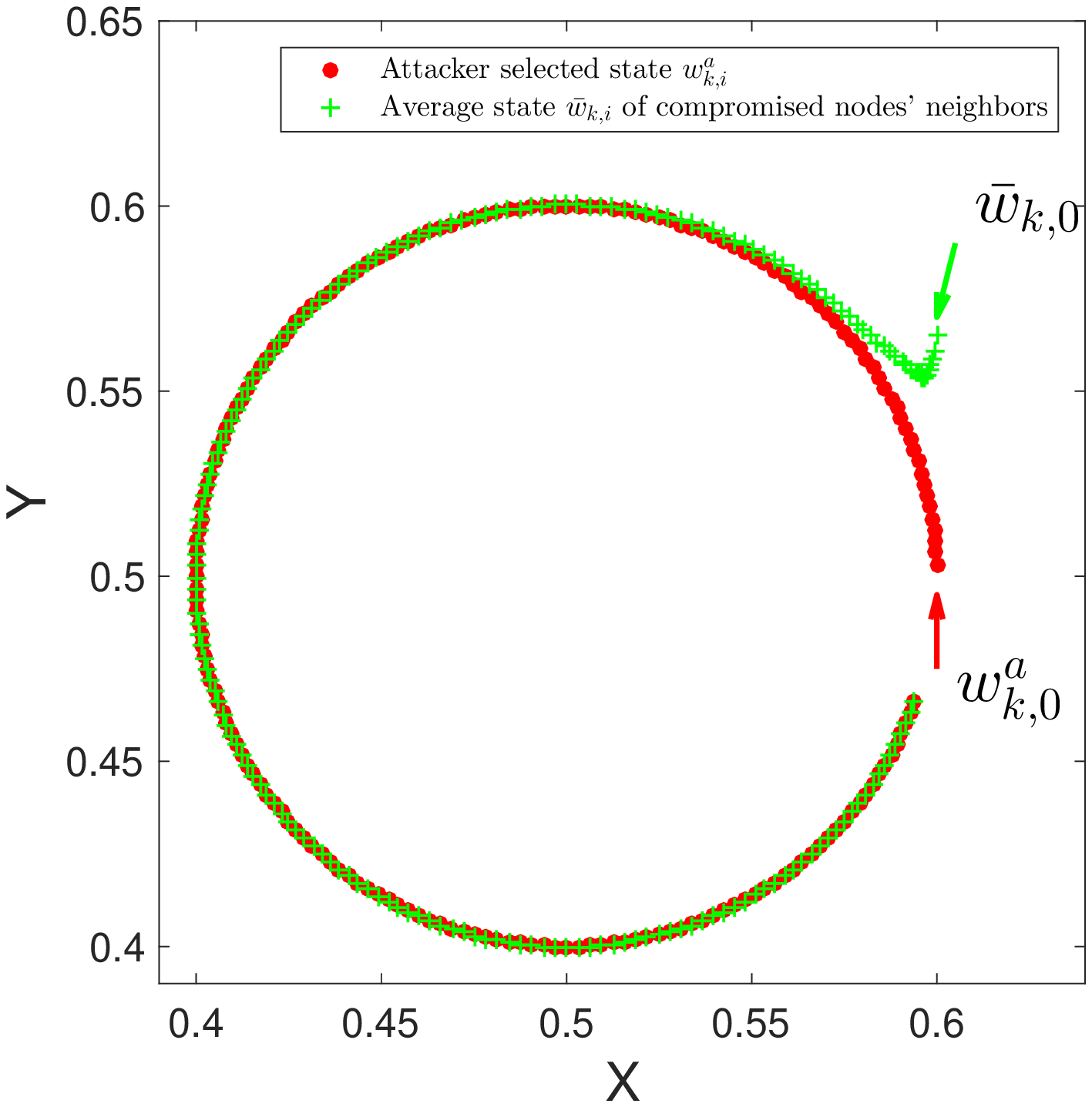}
\caption{Average state dynamics of compromised nodes' neighbors for the first 1900 iterations  (non-stationary targets)}\label{fig: non-stationary_convergence_trends}
\end{minipage}%
\end{figure*} 

\begin{figure*}[htbp]
\vspace{0cm} 
\setlength{\abovecaptionskip}{0.6cm}  
\begin{minipage}[t]{0.33\linewidth}
\centering
\includegraphics[width=0.9\textwidth, trim=1cm 1cm 1cm 1cm]{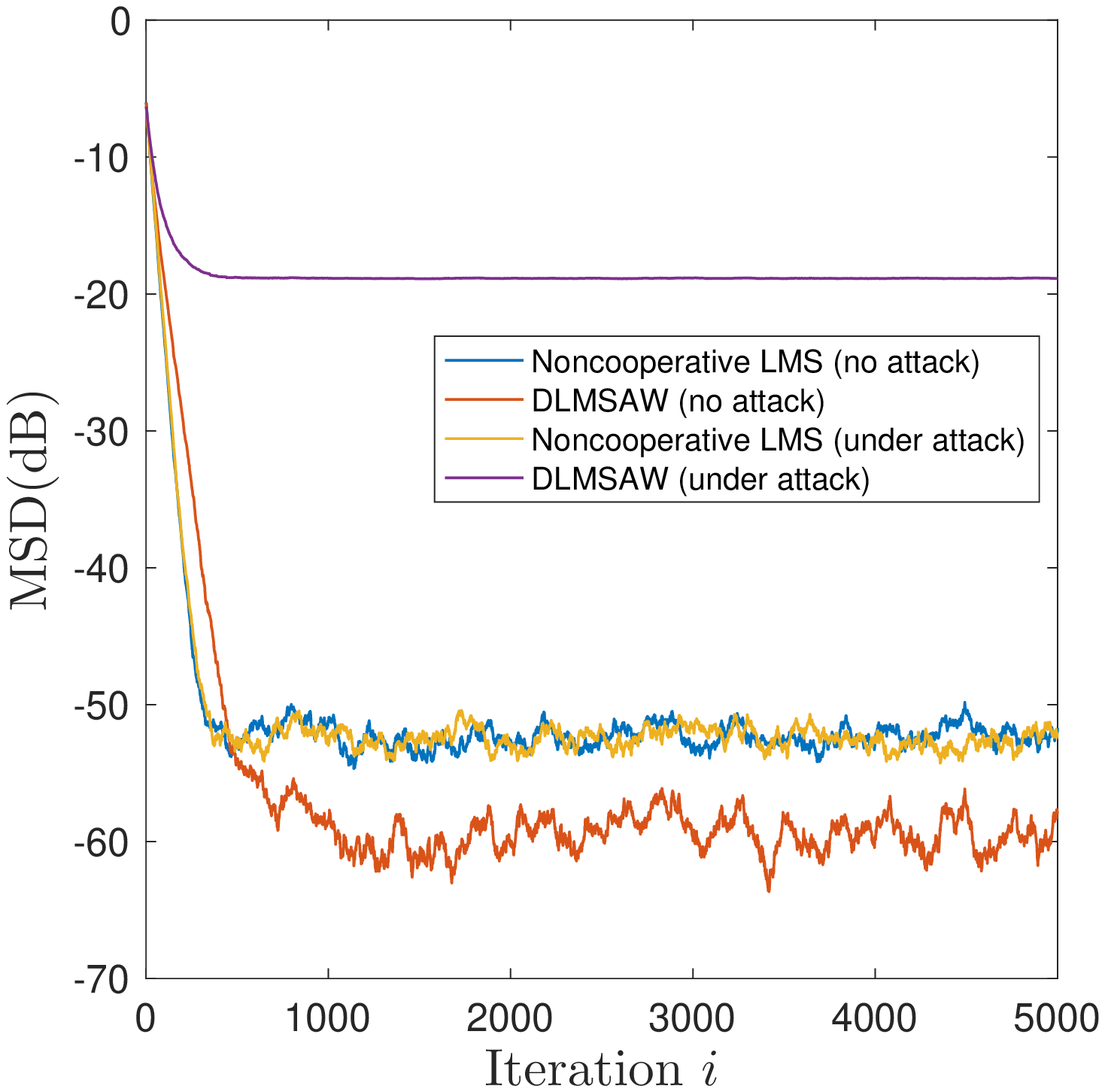}
\caption{MSD level for noncooperative LMS and DLMSAW (non-stationary targets)}\label{fig:non-stationary MSD level}
\end{minipage}%
\begin{minipage}[t]{0.33\linewidth}
\centering
\includegraphics[width=0.9\textwidth, trim=1.5cm 1.5cm 1.5cm 1.5cm]{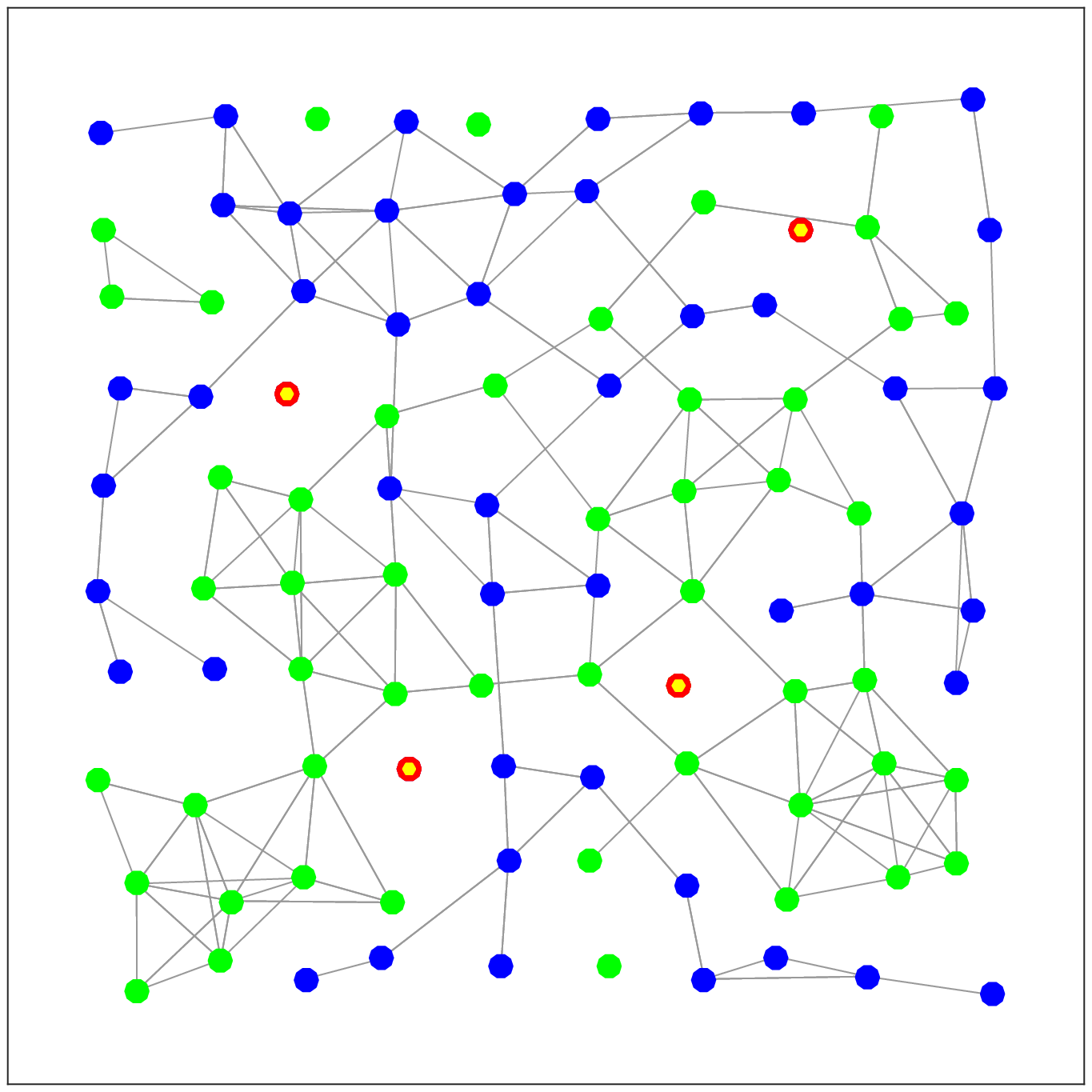}
\caption{Network topology at the end of the simulation (stationary, under attack, F-local resilient, $F = 1$)}\label{fig:network_Flocal}
\end{minipage}%
\begin{minipage}[t]{0.33\linewidth}
\centering
\includegraphics[width=0.9\textwidth, trim=1cm 1cm 1cm 1cm]{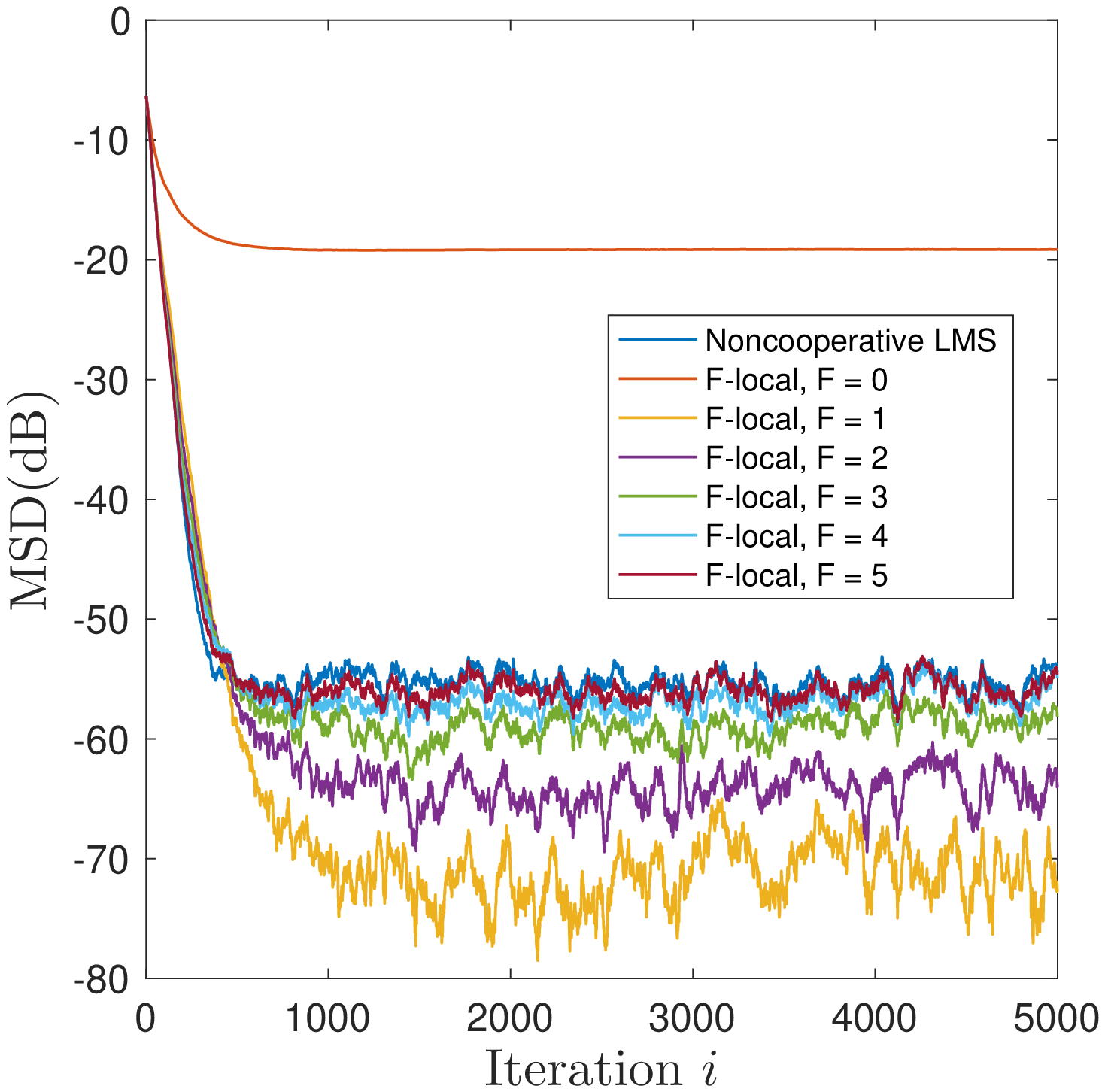}
\caption{MSD for noncooperative LMS and F-local resilient algorithm (stationary, under attack)}\label{fig: MSD_Flocal}
\end{minipage}%
\end{figure*}

\section{Related Work}
Many distributed algorithms are vulnerable to cyber attacks. The existence of an adversarial agent may prevent the algorithm from performing the desired task. 
Two main strategies to address distributed estimation/optimization problems are based either on consensus or on diffusion. 
Resilience of consensus-based distributed algorithms in the presence of cyber attacks
has received considerable attention. 
In particular, the approaches presented in \cite{DBLP:journals/tac/PasqualettiBB12, DBLP:journals/jsac/LeBlancZKS13, DBLP:conf/hicons/LeBlancH14} consider the consensus 
problem for scalar parameters in the presence of attackers, and resilience is achieved 
by leveraging high connectivity. Resilience has been studied also for triangular networks
for distributed robotic applications \cite{TriangularNetworks}.  The approach presented in
\cite{6900100} incorporates ``trusted nodes'' that cannot be attacked to improve the
resilience of distributed consensus.
Typical approaches usually assume Byzantine faults and consider that the goal of the attacker
is to disrupt the convergence (stability) of the distributed algorithm. In contrast, this 
work focuses on attacks that do not disrupt convergence but drive the normal agents to
converge to states selected by the attacker. 

Resilience of diffusion-based distributed algorithms has been studied in \cite{6232902} and \cite{journals/spm/SayedTCZT13}. The main idea is to consider the presence of intruders and use adaptive weights to counteract the attacks. This is an effective measure and has been applied to multi-task networks and distributed clustering problems \cite{6232902}. Several variants focusing on adaptive weights applied to multi-task networks can be found in \cite{7060710,  6845334, 7065284}. The approach presented in \cite{yuanchen/AdversaryDetection} proposes an Flag Raising Distributed Estimation algorithm where a normal agent raises an alarm if any of its neighbors' estimate deviates from its own estimate beyond a given threshold. This is similar to assigning adaptive weights to neighbors. 
Although adaptive weights provide some degree of resilience to attacks, 
we have shown in this work that adaptive weights may introduce vulnerabilities that allow
deception attacks. 

Finally, there has been considerable work on applications of diffusion algorithms 
that include spectrum sensing in cognitive networks \cite{7086338}, target localization \cite{targetLocalization}, distributed clustering \cite{6232902},  biologically inspired designs \cite{mobileAdaptiveNetworks}. Although our approach can be used for resilience of various 
applications, we focus on multi-target localization \cite{6854652}.

\section{Conclusions}
In this paper, we studied distributed diffusion for multi-task networks and investigated vulnerabilities  introduced by adaptive weights. We proposed attack models  
that can drive normal agents to  any state selected by the attacker, for both stationary  and non-stationary estimation. We then developed a resilient distributed diffusion algorithm for counteracting message falsification byzantine attack aiming at making normal agents converge 
to a selected state. Finally, we evaluate our results by  stationary and non-stationary 
multi-target localization.

\section{Acknowledgments}
This work is supported in part by the National Science Foundation (CNS-1238959), the Air Force Research Laboratory (FA 8750-14-2-0180), and by NIST (70NANB17H266). Any options, findings, and conclusions or recommendations expressed in this material are those of the author(s) and do not necessarily reflect the views of AFRL, NSF and NIST.

%
\bibliographystyle{unsrt}
\bibliography{sigproc}  
%
%

\newpage

\appendix
\centerline{{\textit{Proof of Lemma 1}}}
Assume $M$ is the normal neighbors set of $k$ not connected to $a$, and $N$ is the normal neighbors set of $k$ including $k$ itself connected to $a$.
At iteration $i$:
\begin{equation}\label{eq: lemma1 1}
    \bm{w}_{k,i} = \sum_{l \in \{M, N\}}a_{lk}(i)\bm{\psi}_{l,i} + a_{ak}(i)\bm{\psi}_{a,i}
\end{equation}
\noindent For $i > i_c$, the following equations hold for agents $l \in M$:
\begin{equation*}
        \bm{w}_{l,i-1} \approx w_l^0, \qquad e_{l}(i) = \bm{d}_{l}(i)-\bm{u}_{l,i}\bm{w}_{l,i-1} \approx 0
\end{equation*}
\begin{equation*}
        \bm{\psi}_{l,i} = \bm{w}_{l,i-1} + \mu_l \bm{u}_{l,i}^* e_{l}(i) \approx \bm{w}_{l,i-1} \approx w_l^0
\end{equation*}
Assuming the attack succeeds, and all $l \in N$ will be driven to converge to $w_l^a$, we have:
\begin{equation*}
        \bm{w}_{l,i-1} \approx w_l^a, \quad e_{l}(i) = \bm{d}_{l}(i)-\bm{u}_{l,i}\bm{w}_{l,i-1} \neq 0, \quad \bm{\psi}_{l,i} \neq \bm{w}_{l,i-1}  
\end{equation*}
As a result, for $i > i_c$, equation \eqref{eq: lemma1 1} can be written as:
\begin{equation*}
    \bm{w}_{k,i} = \sum_{l \in M}a_{lk}(i)w_l^0 + \sum_{l \in N}a_{lk}(i)(\bm{w}_{l,i} + \mu_{l} \bm{u}_{l,i}^* e_{l}(i))+ a_{ak}(i)\bm{\psi}_{a,i}
\end{equation*}
As observed by the above equation, for $i > i_c$, $\bm{w}_{k,i}$ is determined by multiple variables but the attacker can only manipulate the value of $\bm{\psi}_{a,i}$ and $a_{ak}(i)$, and thus indirectly manipulate $a_{lk}(i)$ for $l \in \{M, N\}$. Assume the attack succeeds and thus $\exists i_c$, s.t. $\forall i > i_c$, $\| \bm{w}_{k,i} - w^a_{k,i}\| < \epsilon$, for some small value $\epsilon > 0$. As a result, for $i > i_c$, node $a$ must make $a_{lk}(i) \rightarrow 0$ for $l \in \{M, N\}$. If not, $\bm{w}_{k,i}$ will be determined by some uncontrollable variables and cannot stay at the specific state selected by the attacker. Thus, for stationary state estimation, we finally get $\forall i > i_c$, $\forall \epsilon > 0$, $a_{ak}(i) > 1 - \epsilon$.

It's easy to verify that by manipulating $\psi_{a,k,i} = \bm{w}_{k,i-1}$ for  each $i$, $\forall \epsilon > 0$, $a_{ak}(i) > 1 - \epsilon$  holds at a certain point. Yet one could easily find the compromised node cannot achieve its goal of making node $k$ to converge to a selected state by such strategy.
The reason is when $k$ aggregates its neighbors' estimation at each iteration $i$, it actually updates its state $\bm{w}_{k,i}$ to the message it receives from $a$. Since this message is equal to $\bm{w}_{k,i-1}$, $\bm{w}_{k,i}$ does not change from $\bm{w}_{k,i-1}$. To conclude, once  $\forall \epsilon > 0$, $a_{ak}(i) > 1 - \epsilon$, node $k$ does not change its state.
Therefore, to make $k$'s state change, compromised node $a$ should follow a strategy ensuring $\neg(\forall \epsilon > 0, a_{ak}(i) > 1 - \epsilon)$. This condition should hold when the attacker wants node $k$ to change state. For stationary estimation, it applies to the iterations before convergence; and for non-stationary estimation, besides the iterations before convergence, it also applies to that after convergence since it adopts a dynamic model.

Moreover, recall the state update equation \eqref{eq: lemma1 1}, in order to dominate node $k$'s state dynamics, compromised node $a$ must be assigned a sufficient large weight so that to eliminate node $k$'s other neighbors impact on node $k$'s state updates. 
Based on the above facts, the compromised node should follow the following condition to make $k$'s state change: 
\begin{equation*}
    (\forall l \in \mathcal{N}_k \cap l \neq a, a_{ak}(i) \gg a_{lk}(i)) \wedge \neg(\forall \epsilon > 0, a_{ak}(i) > 1 - \epsilon)
\end{equation*}
However, it should be noted that it is tolerant that for some of the iteration towards convergence (or after convergence for non-stationary estimation), the above condition does not hold but attack will also succeed at future point. E.g., $a_{ak}(i) = 1$, at which iteration the state stays unchanged; Or, $a_{ak}(i) \ll 1$, at which iteration the state being assigned a random quantity (can be seen as re-initialization). To conclude, only when the above condition holds, node $k$ makes progress to converge to attacker's selected state. 
As a result, we loose the above condition as that given in \textit{Lemma 1}. 
Also, for stationary estimation, after convergence, $\forall \epsilon > 0$, $a_{ak}(i) > 1 - \epsilon$ should hold since once entering convergence, the state never changes.
\\[10pt]
\centerline{\textit{Proof of Lemma 2}}

We use $\delta_{a,k,i}$ to denote $\|\psi_{a,k,i} - \bm{w}_{k,i-1}\|$, and $\delta_{l,k,i}$ to denote $\|\bm{\psi}_{l,i} - \bm{w}_{k,i-1}\|$, for $l \in \mathcal{N}_k, l \neq a$. At iteration $(i_a + n)$,
\begin{equation*}
\begin{split}
\gamma^2_{ak}(i_a + n) =& (1 - \nu_k)^{n+1} \gamma^2_{ak}(i_a - 1) \\
&+ \nu_k [(1 - \nu_k)^{n} \delta_{a,k,i_a}^2 + (1 - \nu_k)^{n-1} \delta_{a,k,i_a+1}^2 \\
&+ \ldots + (1 - \nu_k) \delta_{a,k,i_a+n-1}^2 +\delta_{a,k,i_a+n}^2]
\end{split}
\end{equation*}
\begin{equation*}
\begin{split}
\gamma^2_{lk}(i_a + n) =& (1 - \nu_k)^{n+1} \gamma^2_{lk}(i_a - 1) \\
&+ \nu_k [(1 - \nu_k)^{n} \delta_{l,k,i_a}^2 + (1 - \nu_k)^{n-1} \delta_{l,k,i_a+1}^2 \\
&+ \ldots + (1 - \nu_k) \delta_{l,k,i_a+n-1}^2 +\delta_{l,k,i_a+n}^2]
\end{split}
\end{equation*}
For large enough $n$, $(1-\nu_k)^{n+1} \rightarrow 0$. Since we assume $\|\psi_{a,k,i} - \bm{w}_{k,i-1}\| \ll \| \bm{\psi}_{l,i} - \bm{w}_{k,i-1}\|$, i.e., $\delta_{a,k,i} \ll \delta_{l,k,i}$, for $i \geq i_a + n$, 
$\gamma^2_{ak}(i) \ll \gamma^2_{lk}(i)$ holds.
Based on equation \eqref{eq: adaptive relative-variance combination rule}, the weight $a_{ak}(i) \gg a_{lk}(i)$.  And since $\|\psi_{a,k,i} - \bm{w}_{k,i-1}\| = 0$ does not always hold, such that $\gamma^2_{ak}(i) = 0$ does not always hold, and as a result, $\forall \epsilon > 0$, $a_{ak}(i) > 1 - \epsilon$ does not always hold. 
And for stationary estimation, for $i > i_c$, $\psi_{a,k,i} = \bm{w}_{k, i}$ renders $\forall \epsilon > 0$, $a_{ak}(i) > 1 - \epsilon$. Thus, the condition in \textit{Lemma 1} can be satisfied by the condition in \textit{Lemma 2}.
\\[10pt]
\centerline{\textit{Proof of Lemma 3}}

Message received by $a$ from $k \in \mathcal{N}_a$ is $\bm{\psi}_{k,i}$. 
To compute $\bm{w}_{k,i-1}$ from $\bm{\psi}_{k,i}$, $k$ can perform the following computation:
\begin{equation*}
\bm{w}_{k,i-1} = \bm{\psi}_{k,i} - \mu_k \bm{u}_{k,i}^* (\bm{d}_k(i) - \bm{u}_{k,i} \bm{w}_{k,i-1})
\end{equation*}
from which it can compute $\bm{w}_{k,i-1}$ as:
\begin{equation*}
\bm{w}_{k,i-1} = \frac{\bm{\psi}_{k,i} - \mu_k \bm{u}_{k,i}^* \bm{d}_k(i)}{1 - \mu_k \bm{u}_{k,i}^* \bm{u}_{k,i}}
\end{equation*}
Assuming that the attacker has knowledge of $\mu_k$, $\bm{d}_k(i)$, and $\bm{u}_{k,i}$, 
the value $\bm{w}_{k,i-1}$ can be computed exactly.
\\[10pt]
\centerline{\textit{Proof of Proposition 1}}

The constraint of $r_{k,i}^a$ is consistent with the condition of \textit{Lemma 2}.
Thus, for $i \geq i_a'$, the state of node $k$ will be attacked as to be:
\begin{equation}\label{eq:attack w}
\begin{split}
\bm{w}_{k,i} &\approx \psi_{a,k,i} = \bm{w}_{k,i-1} + r_{k,i}^a(x_{i} - \bm{w}_{k,i-1}) \\
& =r_{k,i}^a x_{i} + (1 - r_{k,i}^a) \bm{w}_{k,i-1} \\
&(i \geq i_a + n, \text{ subject to } (1-\nu_k)^{n+1} \approx 0)
\end{split}
\end{equation}

let $X_{i}$ be $\bm{w}_{k,i}$, $X_{i-1}$ be $\bm{w}_{k,i-1}$, $A_{i}$ be $r_{k,i}^a x_{i}$, and $B$ be $(1-r_{k,i}^a)$.
Equation \eqref{eq:attack w} turns to:
\begin{equation} \label{convergent proof 1}
X_{i} \approx A_{i} + B X_{i-1}
\end{equation}
Assume $\lim_{i \rightarrow \infty} X_{i-1} = X_{i-1}^0$ and $\lim_{i \rightarrow \infty} X_{i} = X_{i}^0$, then for $i \rightarrow \infty$ we get:
\begin{equation} \label{convergent proof 2}
X_{i}^0 \approx A_{i} + B X_{i-1}^0
\end{equation}
Subtract \eqref{convergent proof 2} from \eqref{convergent proof 1}, we get
\begin{equation*}
X_{i} - X_{i}^0  \approx B (X_{i-1} - X_{i-1}^0)
\end{equation*}
let $\varepsilon_i = X_{i} - X_i^0$, for $i = 0, 1, 2, \ldots$, then $\varepsilon_{i} \approx B \varepsilon_{i-1} \approx B^2 \varepsilon_{i-2} \approx \ldots \approx B^{i} \varepsilon_0$. 
The sufficient and necessary requirement of convergence is
\begin{equation*}
\lim_{i \to \infty} \varepsilon_{i} = 0
\end{equation*}
Or, $\lim_{i \to \infty} B^{i} \varepsilon_{0} = 0$. That is, $\lim_{i \to \infty} B^{i} = 0$. Therefore, we get the sufficient and necessary requirement of convergence is $|B| < 1$.
since $B = 1 - r_{k}^a$, and $r_{k}^a \in (0, 1)$, we get $B \in (0, 1)$. Therefore, $\lim_{i \rightarrow \infty} (X_i - X_i^0) = 0$. The assumption $\lim_{i \rightarrow \infty} X_i = X_i^0$ holds.
Therefore, $X_i$ is convergent to $X_i^0$.

To get the value of $X_i^0$, we need to analyze the following two scenarios: stationary state estimation and non-stationary state estimation, separately.
\subsubsection{Stationary state estimation}
In stationary scenarios, the convergence state is in-dependent of time, i.e., $X_{i}^0 = X_{i-1}^0 = X^0$. Therefore, equation \eqref{convergent proof 2} turns to:
\begin{equation*}
X^0 \approx A_i + B X^0
\end{equation*}
Thus,  $(1-B) X^0 \approx A_i$, $X^0 \approx \frac{A_i}{1-B}$. The convergent point is:
\begin{equation*}
w_{k,i} \approx \frac{r_{k,i}^a x_{i+1}}{1-(1-r_{k,i}^a)} = \frac{r_{k,i}^a w_k^a}{1-(1-r_{k,i}^a)} = w_k^a = w_{k,i}^a, \quad i \rightarrow \infty
\end{equation*}
which realizes the attacker's objective \eqref{eq: objective function}.

\subsubsection{Non-stationary state estimation}
In non-stationary scenarios, we first assume $x_i = w_{k}^a + \theta_{k,i-1}^a$ and later we will show how $\theta_{k,i-1}^a$ turns to  $\theta_{k,i-1}^a + \frac{\Delta \theta_{k,i-1}^a}{r_{k,i}^a}$. 

Assume the convergence point $X_i^0$ is a combination of a time-independent value and a time-dependent value, such that $X_i^0 = X^0 + \rho_i$. Take original values into \eqref{convergent proof 2} and we get:
\begin{equation}\label{combination of w/wo}
X^0 + \rho_{i} \approx r_{k,i}^a (w_k^a + \theta_{k,i-1}^a) + (1-r_{k,i}^a)(X_0 + \rho_{i-1})
\end{equation}
Divided \eqref{combination of w/wo} into the time-independent component and time-dependent component. We get:
\begin{equation*}
X^0 \approx w_k^a, \quad 
\rho_{i} - \rho_{i-1} \approx r_{k,i}^a (\theta_{k,i-1}^a - \rho_{i-1})
\end{equation*}
Let $\Delta \rho_{i-1} = \rho_{i} - \rho_{i-1}$, we get:
\begin{equation}\label{eq:rho}
\rho_{i-1} \approx \theta_{k,i-1}^a - \frac{\Delta \rho_{i-1}}{r_{k,i}^a}
\quad \text{ and } \quad
\rho_{i} \approx \theta_{k,i}^a - \frac{\Delta \rho_{i}}{r_{k,i}^a}
\end{equation}
Thus,
\begin{equation*}
\Delta \rho_{i-1} = \rho_{i} - \rho_{i-1} \approx \theta_{k,i}^a - \theta_{k,i-1}^a - \frac{1}{r_{k,i}^a} (\Delta \rho_{i} - \Delta \rho_{i-1}) 
\end{equation*}
Let $\Delta \theta_{k,i-1}^a = \theta_{k,i}^a - \theta_{k,i-1}^a$ and $\Delta^2 \rho_{i-1} = \Delta \rho_{i} - \Delta \rho_{i-1}$, then 
\begin{equation*}
\Delta \rho_{i-1} \approx \Delta \theta_{k,i-1}^a - \frac{\Delta^2 \rho_{i-1}}{r_{k,i}^a}
\quad \text{ or } \quad 
\Delta \rho_{i} \approx \Delta \theta_{k,i}^a - \frac{\Delta^2 \rho_{i}}{r_{k,i}^a}
\end{equation*}
If we assume $\frac{\Delta^2 \rho_{i}}{r_{k,i}^a} \ll \Delta \theta_{k,i}^a$, then we have $\Delta \rho_{i} \approx \Delta \theta_{k,i}^a$. 
Therefore, \eqref{eq:rho} turns to:
\begin{equation*}
\rho_{i} \approx \theta_{k,i}^a - \frac{\Delta \theta_{k,i}^a}{r_{k,i}^a}
\end{equation*}
Thus, the dynamic convergence point for $k$ is:
\begin{equation*}
w_{k,i} \approx w_k^a + \theta_{k,i}^a - \frac{\Delta \theta_{k,i}^a}{r_{k,i}^a}, \qquad i \rightarrow \infty
\end{equation*}
This means when sending ${\psi}_{a,k,i} = \bm{w}_{k,i-1} + r_{k,i}^a (w_k^a + \theta_{k,i-1}^a - \bm{w}_{k,i-1})$ as the communication message, the compromised node $a$ can make $k$ converge to $w_k^a + \theta_{k,i}^a - \frac{\Delta \theta_{k,i}^a}{r_{k,i}^a}$. To make agent $k$ converge to a desired state $w_k^a + \Omega_{k,i}^a$, we assume the message being sent is:
\begin{equation*}
{\psi}_{a,k,i} = \bm{w}_{k,i-1} + r_{k,i}^a (w_k^a + m_{i-1} - \bm{w}_{k,i-1})
\end{equation*}
And the corresponding convergence point will be 
$ w_k^a + m_{i} - \frac{\Delta m_{i}}{r_{k,i}^a}$. We want the following equation holds:
\begin{equation}\label{solve convergence point}
w_k^a + m_{i} - \frac{\Delta m_{i}}{r_{k,i}^a} = w_k^a + \Omega_{k,i}^a
\end{equation}
Assuming $\Delta^2 m_i \rightarrow 0$, the solution of \eqref{solve convergence point} is: $m_i = \Omega_{k,i}^a + \frac{\Delta \Omega_{k,i}^a}{r_{k,i}^a}$, meaning to make $k$ converge to a desired state $w_k^a + \Omega_{k,i}^a$, the compromised node $a$ should send communication message:
\begin{equation*}
{\psi}_{a,k,i} = \bm{w}_{k,i-1} + r_{k,i}^a (w_k^a + \Omega_{k,i-1}^a + \frac{\Delta \Omega_{k,i-1}^a}{r_{k,i}^a} - \bm{w}_{k,i-1})
\end{equation*}

Thus, to make $k$ converge to $w_k^a + \theta_{k,i}^a$, the compromised node $a$ should send communication message:
\begin{equation*}
{\psi}_{a,k,i} = \bm{w}_{k,i-1} + r_{k,i}^a (w_k^a + \theta_{k,i-1}^a + \frac{\Delta \theta_{k,i-1}^a}{r_{k,i}^a} - \bm{w}_{k,i-1})
\end{equation*}
The convergence point is:
\begin{equation*}
w_{k,i} = w_k^a + \theta_{k,i}^a = w_{k,i}^a, \qquad i \rightarrow \infty
\end{equation*}
which realizes the attacker's objective \eqref{eq: objective function}.

We can verify the convergence point by putting $x_i = w_k^a + \theta_{k,i-1}^a + \frac{\Delta \theta_{k,i-1}^a}{r_{k,i}^a}, \bm{w}_{k,i} = w_k^a + \theta_{k,i}^a, \bm{w}_{k,i-1} = w_k^a + \theta_{k,i-1}^a$ back into equation \eqref{eq:attack w}, we get:
\begin{equation*}
\begin{split}
w_k^a + \theta_{k,i}^a &\approx r_{k,i}^a (w_k^a + \theta_{k,i-1}^a + \frac{\Delta \theta_{k,i-1}^a}{r_{k,i}^a}) + (1 - r_{k,i}^a)  (w_k^a + \theta_{k,i-1}^a)\\
\theta_{k,i}^a &\approx r_{k,i}^a (\theta_{k,i-1}^a + \frac{\Delta \theta_{k,i-1}^a}{r_{k,i}^a}) + (1 - r_{k,i}^a)  \theta_{k,i-1}^a\\
\theta_{k,i}^a &\approx \theta_{k,i-1}^a + \Delta \theta_{k,i-1}^a 
\end{split}
\end{equation*}
The resulting equation holds, illustrating the validity of the convergence state.

\end{document}